\documentclass[conference]{IEEEtran}
\IEEEoverridecommandlockouts
\usepackage{cite}
\usepackage{amsmath,amssymb,amsfonts}
\usepackage{algorithmic}
\usepackage{algorithm}
\usepackage{graphicx}
\usepackage[caption=false,font=footnotesize]{subfig}
\usepackage{textcomp}
\usepackage[dvipsnames,svgnames,table]{xcolor}
\usepackage{booktabs}
\usepackage{makecell}
\newcommand{\methodname}{HERD}
\newcommand{\method}{\textbf{\texttt{\methodname}}}
\newcommand{\gray}[1]{\textcolor{gray}{#1}}
\newcommand{\venue}[1]{{\scriptsize\textcolor{gray!70!black}{[#1]}}}
\colorlet{bestcell}{red!13}
\colorlet{secondcell}{blue!10}
\newcommand{\best}[1]{\cellcolor{bestcell}\textbf{#1}}
\newcommand{\second}[1]{\cellcolor{secondcell}\textbf{#1}}
\DeclareRobustCommand{\swatch}[1]{{\setlength{\fboxsep}{0pt}\setlength{\fboxrule}{0.4pt}\raisebox{-0.2ex}{\fcolorbox{gray!80}{#1}{\rule{0pt}{1.8ex}\hspace{5ex}}}}}
\newcommand{\band}[1]{\gray{\textit{#1}}}
\definecolor{hdrA}{HTML}{8A1F3D}
\definecolor{hdrB}{HTML}{16479B}
\definecolor{hdrC}{HTML}{15683A}
\definecolor{hdrD}{HTML}{4D4D4D}
\definecolor{hdrE}{HTML}{5B3F8C}
\newcommand{\hd}[2]{\cellcolor{#1}\textcolor{white}{\textbf{#2}}}
\definecolor{oursname}{HTML}{FBE5D0}

\newcommand{\std}[1]{{\scriptsize\textcolor{ForestGreen}{($\pm$#1)}}}
\usepackage{tcolorbox}
\usepackage{ntheorem}
\theoremseparator{:}

\newtheorem{assumption}{Assumption}

\newtheorem{proposition}{Proposition}
\newtheorem{theorem}{Theorem}

\theoremstyle{nonumberplain}
\theorembodyfont{\normalfont}
\newtheorem{proof}{Proof}
\usepackage{hyperref}
\usepackage{balance}

\def\BibTeX{{\rm B\kern-.05em{\sc i\kern-.025em b}\kern-.08em
    T\kern-.1667em\lower.7ex\hbox{E}\kern-.125emX}}
\begin{document}
\bstctlcite{IEEEexample:BSTcontrol}

\title{Tracing Knowledge Lineage for Federated Unlearning under Mutual Distillation}

\author{
\IEEEauthorblockN{Zihao Ding\textsuperscript{1}, Jun Huang\textsuperscript{1}, and Guodong Wang\textsuperscript{2}}
\IEEEauthorblockA{\textsuperscript{1}Department of Electrical Engineering and Computer Science, South Dakota State University, Brookings, SD 57007, USA\\
\textsuperscript{2}Department of Computer Science, Massachusetts College of Liberal Arts, North Adams, MA 01247, USA\\
Email: zihao.ding@jacks.sdstate.edu, jun.huang@sdstate.edu, guodong.wang@mcla.edu}
}

\maketitle

\begin{abstract}
An edge server and nearby devices can train language models of different sizes together by mutual distillation. Federated unlearning (FU) and unlearning for large language models (LLMs) execute a removal request on one model. Under mutual distillation, however, the effect of the removed data has been copied into many students, and their uploads teach it back to the corrected teacher. We propose \underline{H}older-level \underline{E}rasure via \underline{R}eplication-weighted \underline{D}elta-closure (\method{}), which decides which holders must be corrected. \method{} records a knowledge lineage graph from the transfers that the server already performs and weights each edge by its transfer type. It computes how many times the residual at each holder will be copied and selects the cheapest set of holders whose correction bounds the cumulative residual after the request. A transfer guard and a rejoin rule cover the holders left uncorrected or offline. We prove that partial correction suffices only when the copy number of the lineage graph is below one, that averaging alone never reaches one, and that the correction set remains valid under link outages.
On TOFU, \method{} leaves 18\% to 52\% less cumulative residual than the strongest baseline at 6\% to 12\% of its correction cost.
\end{abstract}


\section{Introduction}
\label{sec:intro}

An edge server can host a large language model (LLM), while nearby devices run smaller models that fit their memory and compute~\cite{Wu2026TNSE,Lin2025COMMAG,Wu2026arXivREBASE,Wu2026arXivRouting,Wu2026arXivCrystalMem,Huang2020TWC,Huang2019WCM}. Federated learning (FL) lets these devices learn from local data without sharing it~\cite{Wu2026COMST,Wu2026IoTJ,Wu2026TON,Wu2026Network,Ding2025IPCCC,Ding2026TAAS,Che2027PAKDD}. Models of different architectures cannot be averaged, so the federation exchanges knowledge by mutual distillation~\cite{Itahara2023TMC,He2020NIPS,Fan2025COLING}. The server broadcasts the outputs of its large model, the teacher, on public prompts, and each device trains its small model, the student, on them. The devices upload the outputs of their students, and the teacher learns from them~\cite{Wu2026ICDCSAlpha,Wu2026ARXIV1,Wu2026ARXIVPRISM}.

Federated unlearning (FU)~\cite{Ding2026ICDCS,Wu2026INFOCOM,Ding2026MNET,Pudasaini2026HPSRAgri,Ding2026ARXIVEASE,Ding2026arXivMUTE} and LLM unlearning~\cite{Fan2025NIPS} remove the effect of data whose owner asks for deletion~\cite{Che2026arXiv} or of samples found to be poisoned~\cite{Ding2026ICNC}. Both execute the removal request on one model, typically the global or server model, and treat the correction of that model as complete removal. Mutual distillation breaks this assumption. By the time a request arrives, distillation has copied the effect of the removed data into many students, and the server has cached teacher outputs that carry it. After the teacher is corrected, the students keep uploading outputs shaped by the removed data, and the teacher learns it again. We call this return resurrection. Some students are also offline when the request arrives and rejoin later with their old copies.

Existing work that looks beyond one model does not close this gap. Lineage-based unlearning~\cite{Liu2025TIFS} corrects every descendant of the removed data, but after a few rounds of mutual distillation many holders are descendants, and correcting all of them spends cost on residuals that would fade anyway. Unlearning methods for distillation~\cite{Quan2025NIPS} shard the data in advance, which covers the hop from teacher to student but not the return path. Noising the descendants and distilling them again from the corrected teacher~\cite{Lee2025NIPS} resends the whole distillation set to each of them. None of these methods asks which holders must be corrected to keep the removed knowledge out, or what to do with offline holders.

The mutual distillation loop, however, carries the signals needed to answer this question. Every transfer passes through the server, so the server can record which holder received knowledge from which at no extra cost. Transfers also differ in how they move a residual. Averaging mixes models with weights that sum to one, so it spreads a residual but never enlarges it. Distillation writes the behavior of one model into several students while the sender keeps it, so it copies the residual. Whether a holder must be corrected therefore depends on how many times its residual will be copied, not on how far it lies from the source. The server can estimate these copies from its own records, without sending the removed data to any device.

Motivated by these insights, we design Holder-level Erasure via Replication-weighted Delta-closure (\method{}). \method{} records a knowledge lineage graph from the transfers of each round and weights each edge by its transfer type. It computes a copy weight for every holder, the exposure that one unit of its residual creates through all later transfers, including those that return to the teacher. It then selects the $\delta$-closure, the cheapest set of holders whose correction keeps the cumulative residual after the request below $\delta$. A transfer guard screens later transfers, and a rejoin rule corrects offline holders before their first transfer. The removed data never leave the server.

Our contributions are summarized as follows:
\begin{itemize}
\item We formulate FU under mutual distillation as selecting the cheapest set of holders whose correction bounds the cumulative residual after the request, and we identify the resurrection of the teacher through student uploads.
\item We propose \method{}, which records a typed lineage graph from delivered transfers, selects holders by copy weight, and guards uncorrected and offline holders without querying any device that still holds the removed data.
\item We prove that correcting part of the holders suffices only when the copy number of the lineage graph is below one, that averaging alone never reaches one, and that the selected holder set stays valid under link outages.
\item On TOFU, \method{} leaves 18\% to 52\% less cumulative residual than the strongest baseline at 6\% to 12\% of its correction cost, and every model stays within one point of its retrained counterpart in the forget score.
\end{itemize}

\section{System Model and Problem Formulation}
\label{sec:system}

\begin{figure}[t]
  \centering
  \includegraphics[width=0.95\columnwidth]{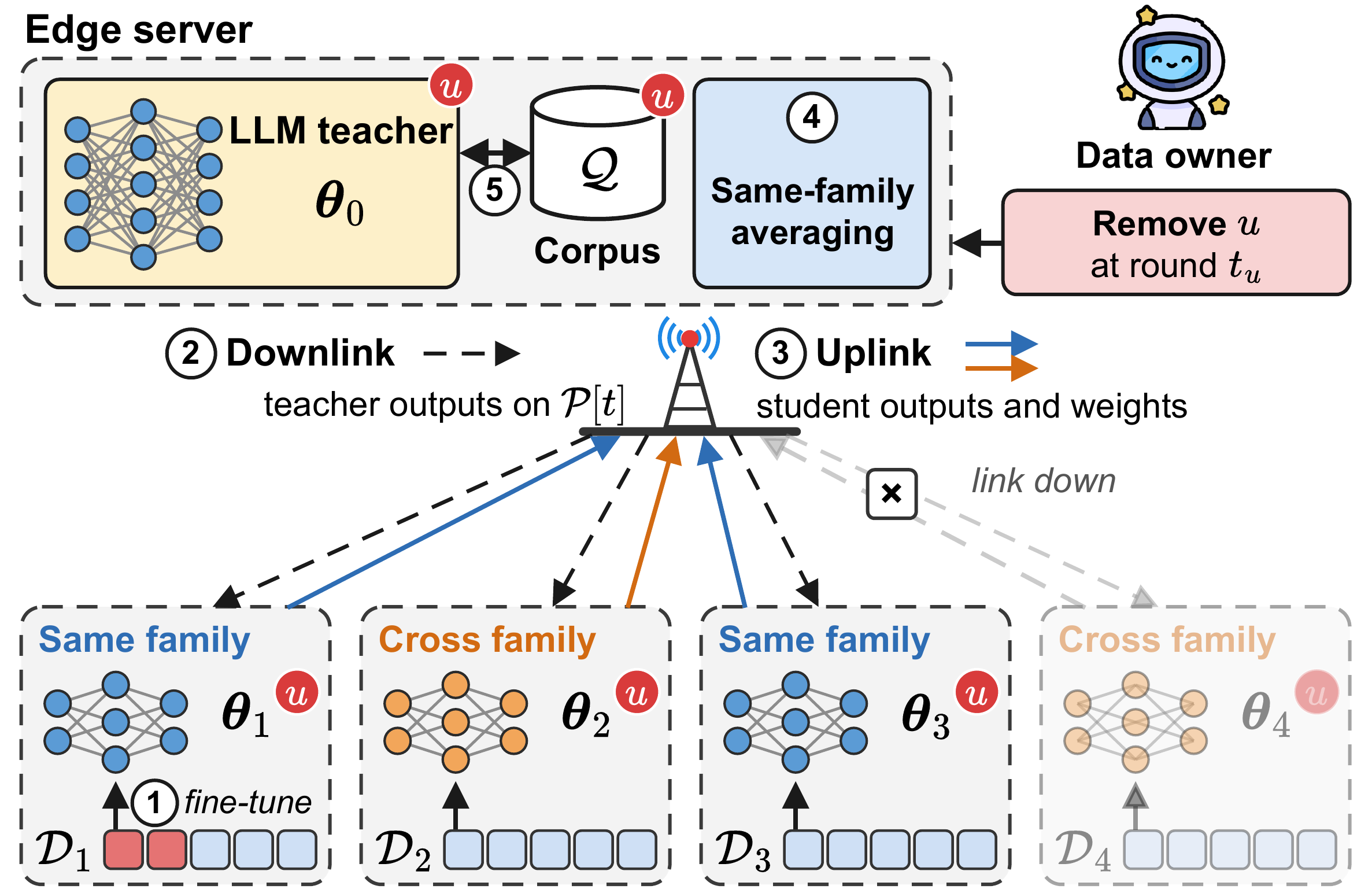}
  \caption{The system model of holder-level unlearning under mutual distillation.}
  \label{fig:system}
\end{figure}

\begin{figure*}[t]
  \centering
  \begin{minipage}[t]{0.48\textwidth}
    \centering
    \includegraphics[width=\linewidth]{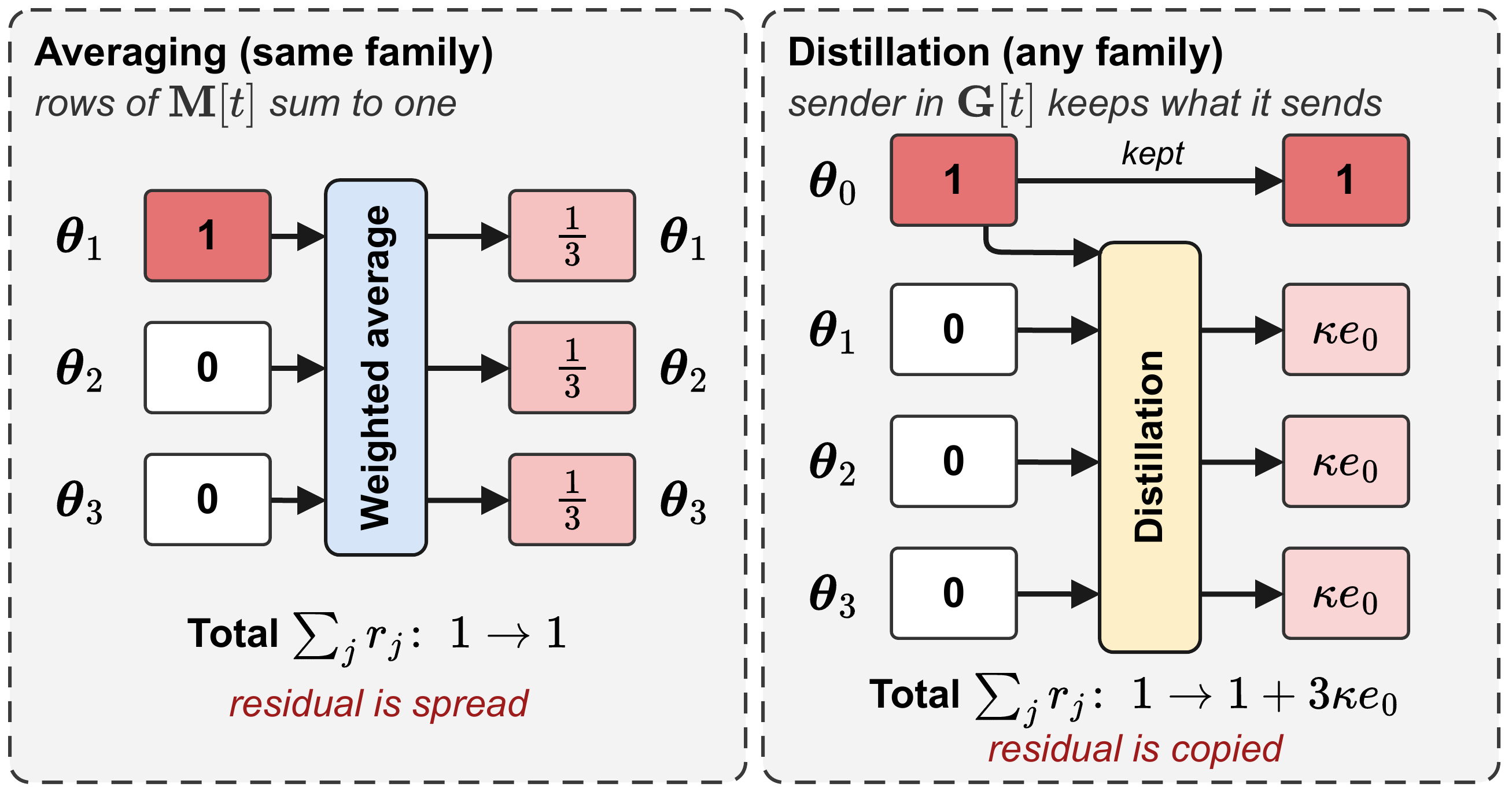}
    \caption{Averaging spreads a residual, while distillation copies it.}
    \label{fig:copy}
  \end{minipage}
  \hfill
  \begin{minipage}[t]{0.48\textwidth}
    \centering
    \includegraphics[width=\linewidth]{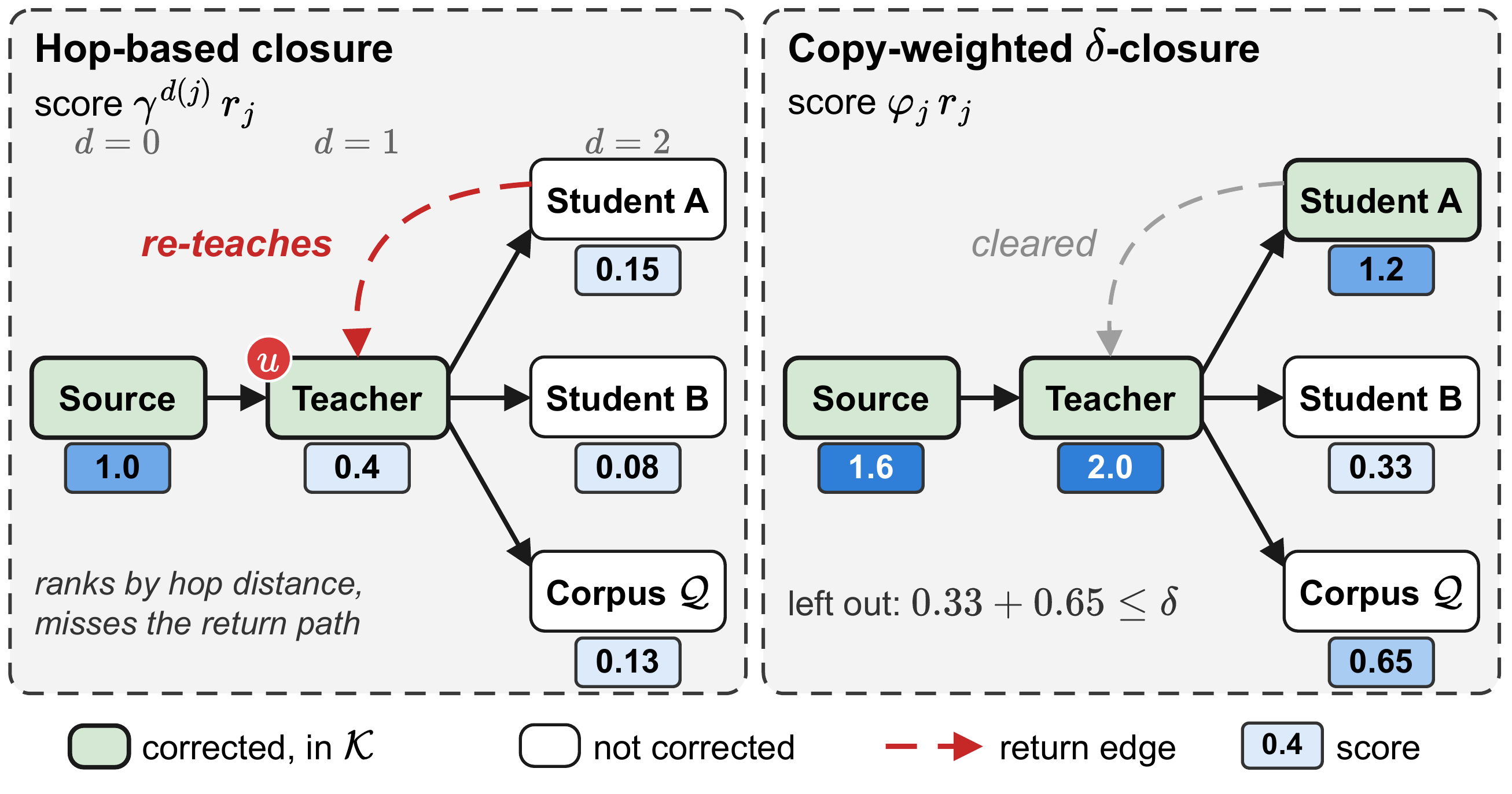}
    \caption{Hop-based and copy-weighted closures with a return edge.}
    \label{fig:closure}
  \end{minipage}
\end{figure*}

\subsection{System Model}
\label{subsec:system-model}

We consider an edge server and $N$ clients connected over wireless links, as illustrated in Fig.~\ref{fig:system}. The server holds a large teacher model $\boldsymbol{\theta}_0$, and client $k$ holds a smaller student model $\boldsymbol{\theta}_k$ that fits its hardware. A student is same-family if it shares the base model of the teacher and cross-family otherwise, and we write $i\sim j$ when two models are of the same family. Client $k$ keeps a local dataset $\mathcal{D}_k$ on board. Since models of different architectures cannot be averaged, the federation exchanges knowledge by mutual distillation on a public prompt set $\mathcal{P}[t]$.

The link of client $k$ has bandwidth $B$, and its rate in round $t$ is
\begin{equation}
    b_k[t]\;=\;B\log_2\!\left(1+\frac{p_k\,h_k[t]}{\sigma^{2}}\right),
    \label{eq:rate}
\end{equation}
where $p_k$ is the transmit power, $\sigma^{2}$ is the noise power, and the channel gain $h_k[t]$ follows Rayleigh fading around the path-loss gain $\bar{h}_k$. The link is down in round $t$ if $b_k[t]$ falls below the rate $b^{\textrm{min}}$ needed to deliver the transfers of a round, which happens with probability
\begin{equation}
    \Pr\{b_k[t]<b^{\textrm{min}}\}\;=\;1-\exp\!\left(-\frac{(2^{b^{\textrm{min}}/B}-1)\,\sigma^{2}}{p_k\,\bar{h}_k}\right).
    \label{eq:outage}
\end{equation}
We denote by $\beta_k$ the probability that the link of client $k$ is down in a round, which also covers the client leaving the coverage of the server. Let $\mathcal{A}[t]$ denote the clients whose links are up in round $t$. In round $t$, every client fine-tunes its student on $\mathcal{D}_k$. The server broadcasts the outputs of $\boldsymbol{\theta}_0[t]$ on $\mathcal{P}[t]$, and each client in $\mathcal{A}[t]$ trains its student on them. Each client in $\mathcal{A}[t]$ then uploads the outputs of its student on $\mathcal{P}[t]$, and the server trains the teacher on them. The server also averages the students of same-family clients in $\mathcal{A}[t]$ with weights proportional to their data sizes. Finally, the server adds the broadcast outputs to a cached corpus $\mathcal{Q}$ and replays part of it in later rounds. The broadcast is one downlink transmission shared by all clients in $\mathcal{A}[t]$, while each upload occupies the link of its own client. A client whose link is down misses the transfers of that round and rejoins with its own student when the link returns.

We call the teacher, each student, and the cached corpus a holder, since each of them can carry learned knowledge. Let $\mathcal{H}$ denote the set of $N+2$ holders, indexed by $i$ and $j$. A removal request names a target $u$ and carries its data $\mathcal{D}_u$ to the server. The target can be all data of one client, a group of entities spread over several clients, or a set of poisoned samples, which correspond to client, class, and sample unlearning. The request arrives at round $t_u$.

\subsection{Problem Formulation}
\label{subsec:problem-formulation}

Conventional FU and LLM unlearning execute a removal request on one model, and the procedure ends once that model is corrected. This does not hold under mutual distillation. By round $t_u$, the effect of $u$ sits in the teacher, in the students that distilled from it, and in the cached corpus. If only the teacher is corrected, the students keep uploading outputs shaped by $u$, and the teacher learns $u$ again.

Let $r_j[t]\in[0,1]$ denote the fraction of the effect of $u$ that remains in holder $j$ at round $t$. It is $0$ for a model trained without $u$ and $1$ for the holder before any removal. Every transfer of round $t$ passes part of a residual from its sender to its receiver, so the residual vector $\mathbf{r}[t]$ evolves as
\begin{equation}
    \mathbf{r}[t+1]\;=\;\mathbf{W}[t]\,\mathbf{r}[t],
    \label{eq:residual}
\end{equation}
where the inheritance matrix $\mathbf{W}[t]$ has entry $w_{ij}[t]$ equal to the fraction of the residual of holder $j$ that holder $i$ inherits in round $t$. A diagonal entry is the part that a holder keeps after its own training, and a client whose link is down has no off-diagonal entry in that round.

At round $t_u$, the server selects a set $\mathcal{K}\subseteq\mathcal{H}$ of holders to correct, which sets their residuals to zero. Correcting holder $j$ costs $c_j=s_j/b_j[t_u]+c_j^{\textrm{cmp}}$, where $s_j$ is the size of the correction sent to $j$ and $c_j^{\textrm{cmp}}$ is its computing time. The teacher and the corpus are corrected on the server, so their $s_j$ is zero. The objective is to minimize the correction cost while bounding the cumulative residual after the request:
\begin{equation}
\begin{split}
    (\mathrm{P}):\;\;\min_{\mathcal{K}\subseteq\mathcal{H}}\;\;
    & \sum_{j\in\mathcal{K}} c_j \\
    \mathrm{s.t.}\;\;
    & (\mathrm{C}_1):\;\sum_{t\ge t_u}\mathbf{1}^{\top}\mathbf{r}[t]\;\le\;\delta,
\end{split}
\label{eq:problem}
\end{equation}
where $\mathbf{1}$ is the all-ones vector. The constraint counts every holder in every later round, so a residual that returns through later transfers violates it even if $\mathcal{K}$ clears the teacher at round $t_u$. We call the solution of $(\mathrm{P})$ the $\delta$-closure of $u$.

Problem~$(\mathrm{P})$ cannot be solved directly. The residuals $r_j[t_u]$ are not observable, and measuring them would require querying every holder with $\mathcal{D}_u$, which sends the removed data to the whole federation. Holders whose links are down at $t_u$ can be neither measured nor corrected. The inheritance matrix is not fixed either, because a transfer moves a residual differently depending on its type and on whether $\mathcal{P}[t]$ touches $u$.
\section{\texorpdfstring{\method{}}{HERD}: Holder-level Erasure via Replication-weighted Delta-closure}
\label{sec:method}

To solve $(\mathrm{P})$ in Eq.~\eqref{eq:problem}, we design \underline{\textbf{H}}older-level \underline{\textbf{E}}rasure via \underline{\textbf{R}}eplication-weighted \underline{\textbf{D}}elta-closure (\method{}), whose three modules are described in Sections~\ref{subsec:klg}, \ref{subsec:closure}, and~\ref{subsec:guard}. Fig.~\ref{fig:arch} shows when each module runs. Fig.~\ref{fig:copy} and Fig.~\ref{fig:closure} illustrate the two ideas that the equations alone leave implicit.

\subsection{Typed Knowledge Lineage Graph}
\label{subsec:klg}

This module estimates the inheritance matrix of Eq.~\eqref{eq:residual} from what the server already does. Every transfer in Section~\ref{subsec:system-model} is a transmission that the server either sends or receives. After each round, the server records the broadcast outputs and the clients that acknowledged them, the uploaded outputs, the averaging weights, and the entries added to $\mathcal{Q}$. These records form a knowledge lineage graph whose nodes are the holders and whose edges are the delivered transfers. Keeping the graph adds no transmission, and the server stores outputs and weights rather than model checkpoints.

A transfer moves a residual in a way that depends on its type. We therefore factor the inheritance matrix as
\begin{equation}
    \mathbf{W}[t]\;=\;\boldsymbol{\Lambda}\left(\mathbf{M}[t]+\mathbf{G}[t]\right),
    \label{eq:factor}
\end{equation}
where $\boldsymbol{\Lambda}=\mathrm{diag}(\lambda_j)$ holds the fraction $\lambda_j<1$ that holder $j$ keeps after its own training. The averaging matrix $\mathbf{M}[t]$ has the averaging weights $m_{ij}[t]$ in the row of each averaged student and a unit row for every other holder, so each of its rows sums to one. The distillation matrix $\mathbf{G}[t]\ge\mathbf{0}$ collects the distillation edges, whose sender keeps what it sends.

The weight of a distillation edge depends on how much of $u$ the transmitted outputs carry. Let $e_j[t]\in[0,1]$ denote the excitation of holder $j$ in round $t$, the fraction of its transmitted outputs that carry the content of $u$. After the request, the server obtains $e_j[t]$ by matching its records against $\mathcal{D}_u$. The edge from $j$ to $i$ then has the weight
\begin{equation}
    g_{ij}[t]\;=\;
    \begin{cases}
        \kappa^{\textrm{x}}\,e_j[t], & i\not\sim j,\\
        \kappa^{\textrm{s}}\,e_j[t]+\nu, & i\sim j,
    \end{cases}
    \label{eq:edge}
\end{equation}
where $\kappa^{\textrm{x}}$ and $\kappa^{\textrm{s}}$ are the absorption rates of cross-family and same-family distillation. The floor $\nu\ge 0$ covers traits that pass between models of the same family without appearing in the outputs~\cite{Cloud2026Nature}. A replay of $\mathcal{Q}$ is a distillation edge whose sender is the corpus. The server fits $\kappa^{\textrm{x}}$, $\kappa^{\textrm{s}}$, and $\nu$ once with canaries. At the start of training, it plants a few harmless fictitious facts at known holders and measures how far they spread along each edge type. Canaries are harmless, so every holder can be queried about them.

\subsection{Copy-Weighted \texorpdfstring{$\delta$}{Delta}-Closure}
\label{subsec:closure}

This module turns the lineage graph into the set $\mathcal{K}$ of $(\mathrm{P})$. The server first reconstructs $\mathbf{r}[t_u]$ by running Eq.~\eqref{eq:residual} forward with the recorded $\mathbf{W}[t]$, starting from $r_j=1$ at the holders that received $\mathcal{D}_u$. No holder is queried.

After the request, links fail at random, so the server uses the expected inheritance matrix $\bar{\mathbf{W}}$, in which every edge to or from client $k$ is scaled by $1-\beta_k$. The copy weights of all holders are
\begin{equation}
    \boldsymbol{\varphi}\;=\;\left(\mathbf{I}-\bar{\mathbf{W}}^{\top}\right)^{-1}\mathbf{1},
    \label{eq:copy}
\end{equation}
where $\mathbf{I}$ is the identity matrix. Entry $\varphi_j$ is the total exposure that one unit of residual at holder $j$ creates through all later transfers, including those that return to the teacher. When the spectral radius $\rho(\bar{\mathbf{W}})$ is below one, \textbf{Theorem}~\ref{thm:threshold} in Section~\ref{sec:analysis} turns $(\mathrm{C}_1)$ into
\begin{equation}
    \sum_{j\in\mathcal{H}\setminus\mathcal{K}}\varphi_j\,r_j[t_u]\;\le\;\delta.
    \label{eq:closure}
\end{equation}

A closure based on hop distance discounts a holder by its number of hops from the source. Under mutual distillation, a student one hop from the teacher can teach the teacher again, so the hop count of a holder is not unique and the discount misses the return path. Fig.~\ref{fig:closure} contrasts the two closures on a graph with one return edge.

The left side of Eq.~\eqref{eq:closure} is linear in the set of uncorrected holders. With unit costs, adding holders to $\mathcal{K}$ in decreasing order of $\varphi_j r_j[t_u]$ until Eq.~\eqref{eq:closure} holds gives the smallest closure. With general costs, $(\mathrm{P})$ becomes a covering knapsack over $N+2$ items, which the server solves by dynamic programming over a discretized exposure with the guarantee of \textbf{Proposition}~\ref{prop:select}.

The server corrects the teacher with $\mathcal{D}_u$ and removes the corpus entries flagged by the excitation. To correct student $j$, it sends the outputs of the corrected teacher on the recorded prompts that carried $u$, and the student distills on them. These outputs make up $s_j$ in $c_j$, so a student behind a slow link costs more to correct. The data $\mathcal{D}_u$ never leave the server.

\subsection{Transfer Guard and Rejoin}
\label{subsec:guard}

This module covers the holders that $\mathcal{K}$ leaves out. After the request, the server screens every transfer with the matching used for $e_j[t]$, including its own broadcasts, the replays of $\mathcal{Q}$, and the uploads, and it drops the outputs that carry $u$. The guard thus sets $e_j[t]$ to zero in Eq.~\eqref{eq:edge}. It needs no transmission, since every transfer already passes through the server. Only the floor $\nu$ of same-family edges survives, which is why same-family holders with large copy weights enter $\mathcal{K}$ first.

A client whose link is down at $t_u$ can be neither corrected nor guarded. The lineage graph keeps the round of its last transfer, from which Eq.~\eqref{eq:residual} gives its residual. When its link returns, the server recomputes its copy weight with the current $\bar{\mathbf{W}}$. It corrects the client before accepting its first upload if $\varphi_j r_j$ exceeds the remaining budget, and admits it under the guard otherwise. \textbf{Proposition}~\ref{prop:offline} gives the condition under which $\mathcal{K}$ stays valid meanwhile.

Beyond mutual distillation, \method{} stores the transmitted outputs of each round, screens the transfers after a request, and inverts one $(N+2)\times(N+2)$ matrix per request. It adds no transmission other than the corrections of the students in $\mathcal{K}$.

\begin{figure*}[t]
  \centering
  \begin{minipage}[t]{0.48\textwidth}
    \centering
    \includegraphics[width=\linewidth]{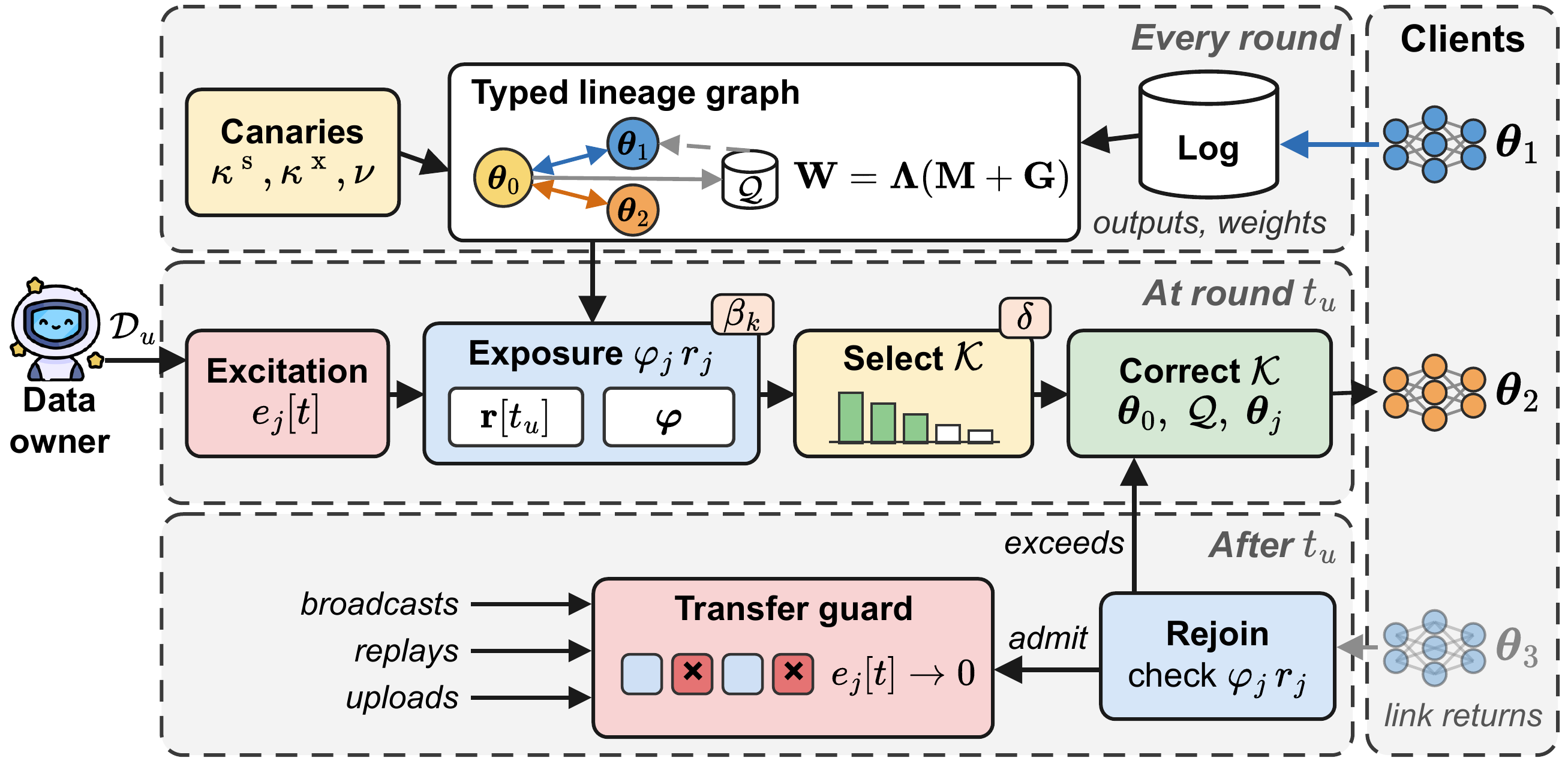}
    \caption{The architecture of \method{} on the edge server.}
    \label{fig:arch}
  \end{minipage}
  \hfill
  \begin{minipage}[t]{0.48\textwidth}
    \centering
    \includegraphics[width=\linewidth]{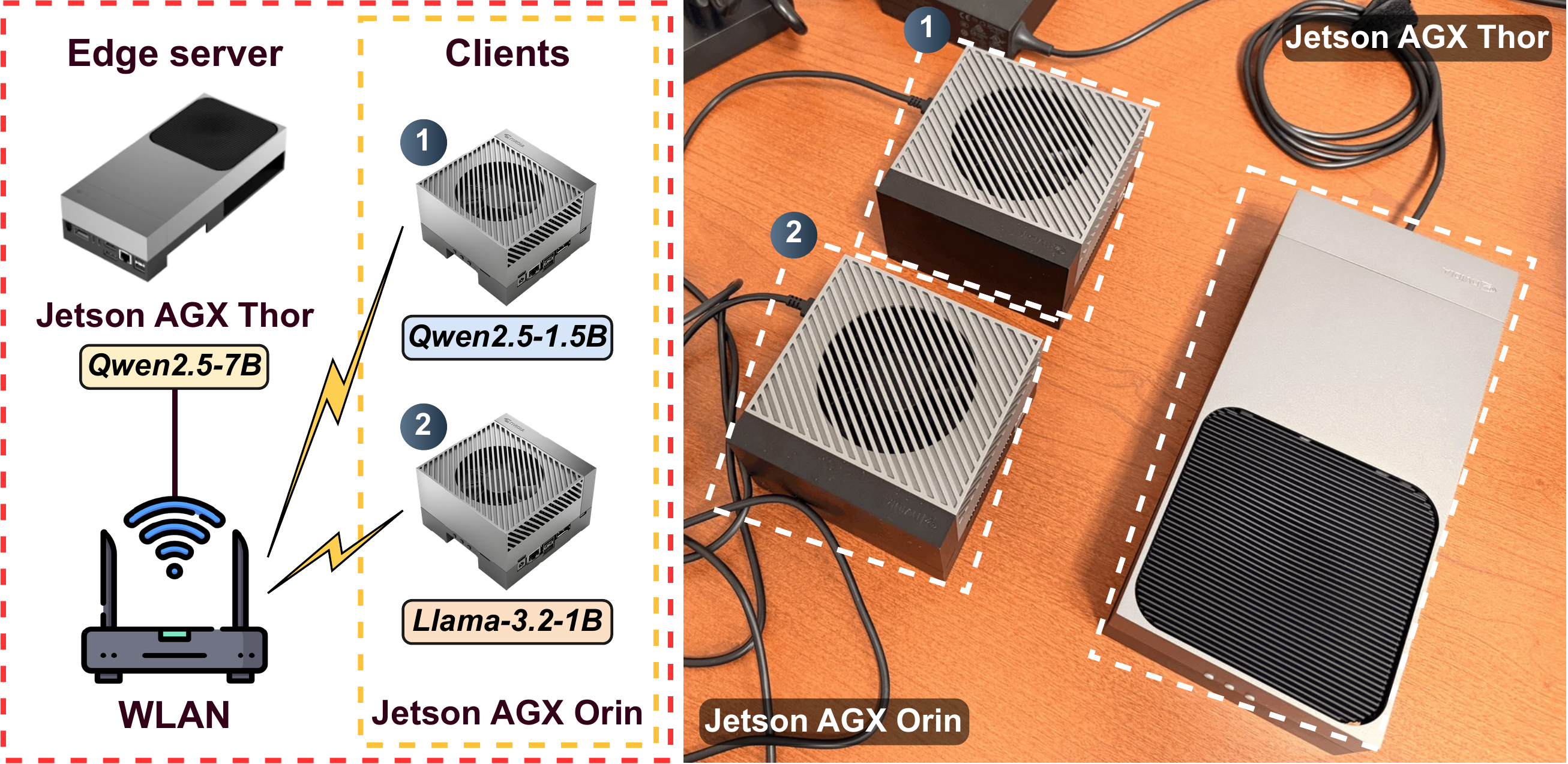}
    \caption{Hardware testbed of Jetson boards over a WLAN.}
    \label{fig:testbed}
  \end{minipage}
\end{figure*}

\begin{table*}[!t]
\centering
\setlength{\tabcolsep}{4pt}
\setlength{\aboverulesep}{0pt}
\setlength{\belowrulesep}{0pt}
\setlength{\extrarowheight}{0pt}
\renewcommand{\arraystretch}{0.88}
\caption{Unlearning, resurrection, and correction cost per method and unlearning type. $\mathrm{FRL}$, $\mathrm{MU}$, and $\mathrm{MIA}$ are in \%, $\mathrm{FQ}$ is a p-value, $\mathrm{Res}$ is in percentage points, $\mathrm{Exp}$ in holder-rounds, and the cost in seconds. Each value is the mean over five random seeds with the standard deviation in parentheses. The student columns average over the same-family and over the cross-family students. The testbed columns come from the Jetson boards, with the time in seconds, the traffic in MB, and the energy in kJ that the correction adds. $\mathrm{FRL}$ and $\mathrm{MIA}$ are best when closest to Retrain of the same unlearning type. A \swatch{bestcell} cell marks the best value and a \swatch{secondcell} cell the second best within each unlearning type, excluding Retrain.}
\label{tab:main}
\resizebox{\linewidth}{!}{\begin{tabular}{l|cccc|cc|cc|cccc|ccc}
\Xhline{1.2pt}
 & \multicolumn{4}{c|}{\hd{hdrA}{Teacher}} & \multicolumn{2}{c|}{\hd{hdrB}{Same family}} & \multicolumn{2}{c|}{\hd{hdrE}{Cross family}} & \multicolumn{4}{c|}{\hd{hdrC}{Federation}} & \multicolumn{3}{c}{\hd{hdrD}{Testbed}}\\
\textbf{Method} & FRL & FQ$\uparrow$ & MU$\uparrow$ & MIA & FRL & MU$\uparrow$ & FRL & MU$\uparrow$ & Res$\downarrow$ & Exp$\downarrow$ & $|\mathcal{K}|$ & Cost$\downarrow$ & Time$\downarrow$ & Traffic$\downarrow$ & Energy$\downarrow$\\
\Xhline{1.2pt}
\multicolumn{16}{l}{\band{(a) Client unlearning}}\\
\Xhline{0.6pt}
SimNPO \venue{NeurIPS'25}~\cite{Fan2025NIPS} & 29.19\std{1.61} & 0.147\std{0.127} & 13.92\std{0.65} & 57.13\std{3.82} & 32.01\std{1.88} & \best{31.36\std{1.59}} & 29.84\std{1.22} & \second{33.71\std{0.88}} & 4.72\std{1.35} & 171.86\std{31.47} & 1 & \best{20.3\std{2.3}} & \best{35.7\std{0.9}} & \best{0.26\std{0.02}} & \best{2.42\std{0.16}}\\
FIUn \venue{TIFS'25}~\cite{Liu2025TIFS} & 26.94\std{0.95} & 0.405\std{0.208} & 12.88\std{0.53} & 53.38\std{4.30} & 27.29\std{0.73} & 26.89\std{1.70} & 26.89\std{0.83} & 31.53\std{1.48} & 0.95\std{0.82} & 26.63\std{8.32} & 34.6\std{2.88} & 228.2\std{26.0} & 308.0\std{33.4} & 43.03\std{2.25} & 33.78\std{3.79}\\
DU \venue{INFOCOM'26}~\cite{Wu2026INFOCOM} & 25.85\std{0.81} & 0.576\std{0.297} & 13.35\std{0.43} & \second{50.25\std{1.80}} & 26.84\std{0.22} & 27.59\std{0.92} & 26.79\std{0.70} & 29.42\std{1.12} & \second{0.30\std{0.16}} & 17.94\std{3.85} & 34.6\std{2.88} & 229.9\std{5.0} & 219.7\std{15.0} & 81.55\std{4.51} & 25.39\std{2.17}\\
UNDO \venue{NeurIPS'25}~\cite{Lee2025NIPS} & \best{25.17\std{0.76}} & \best{0.817\std{0.219}} & 13.70\std{0.72} & \best{50.13\std{2.77}} & \second{26.56\std{0.42}} & 28.76\std{1.64} & 26.85\std{0.78} & 31.61\std{0.85} & \best{0.00\std{0.00}} & \second{16.42\std{3.28}} & 34.6\std{2.88} & 1371.8\std{45.2} & 1496.0\std{56.7} & \second{6.91\std{0.23}} & 175.08\std{7.56}\\
PURGE \venue{NeurIPS'25}~\cite{Quan2025NIPS} & 25.90\std{0.75} & 0.580\std{0.301} & \best{15.28\std{0.57}} & 50.50\std{1.95} & 26.95\std{0.53} & 29.82\std{1.58} & \second{26.73\std{0.67}} & \best{33.82\std{1.32}} & 0.52\std{0.15} & 43.79\std{3.69} & 11 & 449.2\std{40.2} & 575.8\std{39.2} & 58.92\std{2.72} & 74.87\std{5.45}\\
\Xhline{0.8pt}
\cellcolor{oursname}\textit{\textbf{\methodname{}}} & \second{25.27\std{0.72}} & \second{0.614\std{0.338}} & \second{14.74\std{0.51}} & 50.75\std{2.14} & \best{26.12\std{0.35}} & \second{30.51\std{3.02}} & \best{26.50\std{0.69}} & 32.87\std{1.66} & 0.48\std{0.15} & \best{13.49\std{1.78}} & 16.2\std{1.64} & \second{82.6\std{5.5}} & \second{122.2\std{6.1}} & 9.93\std{0.99} & \second{10.68\std{0.68}}\\
\Xhline{0.8pt}
\textit{Retrain} & 24.66\std{0.76} & 1.000 & 14.78\std{0.61} & 49.25\std{2.27} & 25.91\std{0.34} & 30.23\std{1.96} & 26.36\std{0.70} & 32.84\std{1.19} & 0.08\std{0.08} & 0.00 & - & 4034.6\std{281.8} & 5422.5\std{368.9} & 1683.97\std{25.54} & 636.61\std{45.96}\\
\Xhline{0.8pt}
\multicolumn{16}{l}{\band{(b) Class unlearning}}\\
\Xhline{0.6pt}
SimNPO \venue{NeurIPS'25}~\cite{Fan2025NIPS} & 32.15\std{3.70} & 0.008\std{0.013} & 14.33\std{0.48} & 60.30\std{4.73} & 33.65\std{3.87} & \best{31.15\std{0.92}} & 32.97\std{2.23} & 32.83\std{1.55} & 6.39\std{3.42} & 278.20\std{63.73} & 1 & \best{101.8\std{10.0}} & \best{120.8\std{4.7}} & \best{0.63\std{0.03}} & \best{7.91\std{0.43}}\\
FIUn \venue{TIFS'25}~\cite{Liu2025TIFS} & 29.41\std{1.19} & 0.049\std{0.072} & 13.61\std{0.41} & 57.40\std{3.26} & 30.44\std{1.29} & 27.59\std{1.14} & 27.04\std{0.59} & 29.85\std{2.35} & 2.03\std{1.16} & 65.12\std{12.39} & 39.8\std{2.77} & 227.4\std{29.6} & 428.8\std{32.2} & 58.23\std{2.97} & 47.19\std{3.06}\\
DU \venue{INFOCOM'26}~\cite{Wu2026INFOCOM} & 28.40\std{1.37} & 0.097\std{0.037} & 14.36\std{0.52} & 55.60\std{1.10} & 27.96\std{0.68} & 28.81\std{0.41} & 27.63\std{0.82} & 30.58\std{1.73} & 1.20\std{0.81} & 49.90\std{10.90} & 39.8\std{2.77} & 192.7\std{10.3} & 272.7\std{14.6} & 83.21\std{2.24} & 30.70\std{2.17}\\
UNDO \venue{NeurIPS'25}~\cite{Lee2025NIPS} & \second{26.66\std{0.83}} & \best{0.555\std{0.257}} & 14.63\std{0.63} & \second{53.00\std{1.67}} & \second{27.50\std{0.28}} & 29.64\std{1.64} & \second{26.85\std{0.58}} & 32.96\std{0.56} & \best{0.37\std{0.15}} & \second{31.55\std{2.84}} & 39.8\std{2.77} & 957.4\std{69.9} & 1699.9\std{129.2} & \second{11.70\std{1.21}} & 200.66\std{11.05}\\
PURGE \venue{NeurIPS'25}~\cite{Quan2025NIPS} & 28.96\std{1.03} & 0.476\std{0.424} & \best{16.12\std{0.68}} & 56.85\std{1.61} & 29.12\std{0.59} & 31.00\std{0.75} & 27.08\std{0.41} & \best{33.63\std{1.29}} & 2.45\std{0.79} & 49.39\std{4.73} & 51 & 2409.1\std{182.8} & 2285.0\std{101.3} & 291.11\std{13.39} & 281.80\std{17.80}\\
\Xhline{0.8pt}
\cellcolor{oursname}\textit{\textbf{\methodname{}}} & \best{26.10\std{0.86}} & \second{0.553\std{0.227}} & \second{15.03\std{0.56}} & \best{52.15\std{1.66}} & \best{26.27\std{0.13}} & \second{31.11\std{0.72}} & \best{26.22\std{0.46}} & \second{33.59\std{1.80}} & \second{0.50\std{0.19}} & \best{25.60\std{2.66}} & 24.6\std{1.95} & \second{110.2\std{8.7}} & \second{187.2\std{7.7}} & 14.50\std{0.43} & \second{17.46\std{0.64}}\\
\Xhline{0.8pt}
\textit{Retrain} & 25.58\std{0.89} & 1.000 & 15.47\std{0.70} & 51.65\std{1.91} & 25.97\std{0.12} & 30.86\std{1.03} & 26.03\std{0.43} & 33.11\std{1.25} & 0.20\std{0.15} & 0.00 & - & 3717.7\std{253.9} & 5188.3\std{270.1} & 1832.84\std{61.97} & 637.03\std{44.27}\\
\Xhline{0.8pt}
\multicolumn{16}{l}{\band{(c) Sample unlearning}}\\
\Xhline{0.6pt}
SimNPO \venue{NeurIPS'25}~\cite{Fan2025NIPS} & 20.00\std{7.91} & 0.230\std{0.220} & 15.65\std{0.92} & 67.50\std{5.86} & 11.64\std{2.57} & \best{33.18\std{1.37}} & 3.88\std{1.45} & 34.67\std{1.53} & 24.00\std{6.52} & 102.05\std{9.28} & 1 & \best{7.2\std{0.7}} & \best{18.1\std{1.5}} & \best{0.13\std{0.01}} & \best{1.19\std{0.10}}\\
FIUn \venue{TIFS'25}~\cite{Liu2025TIFS} & 6.00\std{4.18} & 0.638\std{0.373} & 14.94\std{0.43} & 57.50\std{1.77} & 2.68\std{0.97} & 29.13\std{1.16} & 1.40\std{0.87} & 32.39\std{0.85} & 9.00\std{4.18} & 16.48\std{5.65} & 26.0\std{3.16} & 199.9\std{13.6} & 202.0\std{10.2} & 33.19\std{0.80} & 21.97\std{1.24}\\
DU \venue{INFOCOM'26}~\cite{Wu2026INFOCOM} & 5.00\std{3.54} & 0.676\std{0.143} & 15.75\std{0.80} & 55.50\std{5.70} & \second{0.52\std{0.23}} & 31.35\std{1.77} & 0.84\std{0.41} & 32.26\std{1.71} & 10.00\std{7.91} & 8.18\std{1.45} & 26.0\std{3.16} & 191.8\std{12.3} & 194.4\std{3.7} & 78.42\std{4.94} & 21.69\std{1.05}\\
UNDO \venue{NeurIPS'25}~\cite{Lee2025NIPS} & \second{2.00\std{2.74}} & \best{0.788\std{0.207}} & 15.89\std{0.67} & \second{53.00\std{3.71}} & 1.60\std{0.65} & 31.40\std{1.30} & \second{0.32\std{0.23}} & 33.54\std{1.80} & \second{8.00\std{5.70}} & \second{8.02\std{1.06}} & 26.0\std{3.16} & 1035.6\std{55.1} & 1006.5\std{52.1} & \second{4.19\std{0.19}} & 119.92\std{8.34}\\
PURGE \venue{NeurIPS'25}~\cite{Quan2025NIPS} & 3.00\std{4.47} & 0.659\std{0.252} & \best{17.58\std{0.68}} & 54.50\std{2.09} & 1.56\std{0.52} & 32.04\std{1.33} & 1.44\std{0.54} & \best{35.11\std{1.26}} & 9.00\std{2.24} & 21.80\std{1.29} & 11 & 440.9\std{29.2} & 640.2\std{31.8} & 61.52\std{2.44} & 80.09\std{5.49}\\
\Xhline{0.8pt}
\cellcolor{oursname}\textit{\textbf{\methodname{}}} & \best{1.00\std{2.24}} & \second{0.758\std{0.181}} & \second{17.18\std{0.82}} & \best{51.50\std{2.85}} & \best{0.28\std{0.30}} & \second{32.76\std{0.63}} & \best{0.16\std{0.17}} & \second{34.80\std{1.28}} & \best{5.00\std{0.00}} & \best{3.84\std{2.08}} & 12.2\std{1.79} & \second{68.9\std{5.6}} & \second{95.6\std{6.7}} & 5.71\std{0.23} & \second{8.77\std{0.23}}\\
\Xhline{0.8pt}
\textit{Retrain} & 0.00 & 1.000 & 17.17\std{0.64} & 52.00\std{2.74} & 0.00 & 31.86\std{1.90} & 0.00 & 33.91\std{1.60} & 0.00 & 0.00 & - & 3616.0\std{164.6} & 5014.5\std{276.4} & 1610.61\std{52.47} & 621.96\std{16.27}\\
\Xhline{1.2pt}
\end{tabular}}
\end{table*}

\section{Theoretical Analysis}
\label{sec:analysis}

This section states when the steps of Section~\ref{sec:method} are valid.

\begin{assumption}\label{ass:linear}
After round $t_u$, the matrices $\mathbf{W}[t]$ of Eq.~\eqref{eq:residual} are independent across rounds and of $\mathbf{r}[t]$, and they share the mean $\bar{\mathbf{W}}$, which is nonnegative and irreducible.
\end{assumption}

\begin{theorem}\label{thm:threshold}
Under \textbf{Assumption}~\ref{ass:linear}, let $\mathcal{K}$ be corrected at round $t_u$, and let $\mathbf{r}^{+}$ denote the residuals right after the correction. (i) If $\rho(\bar{\mathbf{W}})<1$, the expected cumulative residual $\mathbb{E}[\sum_{t\ge t_u}\mathbf{1}^{\top}\mathbf{r}[t]]$ equals the left side of Eq.~\eqref{eq:closure}. (ii) If $\rho(\bar{\mathbf{W}})\ge 1$, it is infinite whenever $\mathbf{r}^{+}\neq\mathbf{0}$. (iii) If $\mathbf{G}[t]=\mathbf{0}$ for all $t$, then $\rho(\bar{\mathbf{W}})\le\max_{j}\lambda_j<1$.
\end{theorem}

\begin{proof}
By \textbf{Assumption}~\ref{ass:linear}, $\mathbb{E}[\mathbf{r}[t]]=\bar{\mathbf{W}}^{\,t-t_u}\mathbf{r}^{+}$. For (i), the Neumann series gives $\sum_{t\ge t_u}\bar{\mathbf{W}}^{\,t-t_u}=(\mathbf{I}-\bar{\mathbf{W}})^{-1}$, and $\mathbf{1}^{\top}(\mathbf{I}-\bar{\mathbf{W}})^{-1}\mathbf{r}^{+}=\boldsymbol{\varphi}^{\top}\mathbf{r}^{+}$ by Eq.~\eqref{eq:copy}. For (ii), the Perron--Frobenius theorem gives a positive left eigenvector $\mathbf{y}$ of $\bar{\mathbf{W}}$ for $\rho(\bar{\mathbf{W}})$, so $\mathbf{1}^{\top}\bar{\mathbf{W}}^{\,s}\mathbf{r}^{+}\ge\rho(\bar{\mathbf{W}})^{s}\,\mathbf{y}^{\top}\mathbf{r}^{+}/\max_j y_j$, which does not vanish as $s$ grows. For (iii), every row of $\mathbf{M}[t]$ sums to one, so every row of $\bar{\mathbf{W}}=\boldsymbol{\Lambda}\,\mathbb{E}[\mathbf{M}[t]]$ sums to $\lambda_j$, which bounds $\rho(\bar{\mathbf{W}})$.\end{proof}

According to \textbf{Theorem}~\ref{thm:threshold}, only distillation can make a removal return. Averaging keeps the row sum of holder $i$ at $\lambda_i$, while distillation raises it to $\lambda_i(1+\sum_j\bar{g}_{ij})$, where $\bar{g}_{ij}$ is the mean of $g_{ij}[t]$. A large enough fan-out pushes $\rho(\bar{\mathbf{W}})$ past one, as Fig.~\ref{fig:copy} illustrates. Every holder with any residual must then be corrected, unless the guard cuts the copies first.

\begin{assumption}\label{ass:semantic}
A cross-family student inherits the content of $u$ only through the transmitted outputs.
\end{assumption}

\begin{proposition}\label{prop:guard}
Under \textbf{Assumptions}~\ref{ass:linear} and~\ref{ass:semantic}, the guard gives $\rho(\bar{\mathbf{W}})\le\max_i\lambda_i(1+\bar{\nu}_i)$, where $\bar{\nu}_i$ is $\nu$ times the expected number of same-family distillation edges that enter holder $i$ in a round.
\end{proposition}

\begin{proof}
The guard sets $e_j[t]=0$, so by \textbf{Assumption}~\ref{ass:semantic} and Eq.~\eqref{eq:edge} only the floor $\nu$ of same-family edges remains in $\mathbf{G}[t]$. Row $i$ of $\bar{\mathbf{W}}$ then sums to $\lambda_i(1+\bar{\nu}_i)$, which bounds $\rho(\bar{\mathbf{W}})$.\end{proof}

According to \textbf{Proposition}~\ref{prop:guard}, the guard brings the copy number below one whenever traits spread weakly within a family, which the canaries of Section~\ref{subsec:klg} measure through $\nu$.

\begin{proposition}\label{prop:offline}
Under \textbf{Assumption}~\ref{ass:linear}, let $\rho(\bar{\mathbf{W}})<1$, and let $\mathcal{O}$ denote the clients whose links are down at $t_u$. If every client in $\mathcal{O}$ is corrected before its first transfer after rejoining, then $(\mathrm{C}_1)$ holds whenever
\begin{equation}
    \sum_{j\in\mathcal{H}\setminus(\mathcal{K}\cup\mathcal{O})}\varphi_j\,r_j[t_u]\;+\;\sum_{j\in\mathcal{O}}\frac{r_j[t_u]}{1-\beta_j}\;\le\;\delta.
    \label{eq:offline}
\end{equation}
\end{proposition}

\begin{proof}
Cutting the edges of $\mathcal{O}$ only removes nonnegative entries from $\mathbf{W}[t]$, so the holders outside $\mathcal{K}\cup\mathcal{O}$ contribute at most the first sum by \textbf{Theorem}~\ref{thm:threshold}. A client in $\mathcal{O}$ has no incoming edge before it rejoins, so it keeps at most $r_j[t_u]$ per round. Its link returns after $1/(1-\beta_j)$ rounds in expectation, and it is corrected before its first transfer, which gives the second sum.\end{proof}

According to \textbf{Proposition}~\ref{prop:offline}, the server needs to keep for each offline client only the round of its last transfer, from which its residual follows by Eq.~\eqref{eq:residual}. A poor link enters the budget through $1/(1-\beta_j)$, so a client that is often offline costs budget even when it is never corrected.

\begin{proposition}\label{prop:select}
Let $a_j=\varphi_j\,r_j[t_u]$. (i) With unit costs, adding holders to $\mathcal{K}$ in decreasing order of $a_j$ until Eq.~\eqref{eq:closure} holds gives a closure of minimum size. (ii) With general costs, dynamic programming on every $a_j$ rounded up to a multiple of $\epsilon>0$ gives a $\mathcal{K}$ that satisfies Eq.~\eqref{eq:closure}, costs at most the optimum of $(\mathrm{P})$ with $\delta$ replaced by $\delta-(N+2)\epsilon$, and takes $O((N+2)\,\delta/\epsilon)$ time.
\end{proposition}

\begin{proof}
Let $\mathcal{U}=\mathcal{H}\setminus\mathcal{K}$. Eq.~\eqref{eq:closure} reads $\sum_{j\in\mathcal{U}}a_j\le\delta$, and minimizing the cost of $\mathcal{K}$ maximizes $\sum_{j\in\mathcal{U}}c_j$. For (i), if a feasible $\mathcal{U}$ has $m$ holders, the $m$ smallest $a_j$ also sum to at most $\delta$, so the decreasing order leaves at least $m$ holders uncorrected. For (ii), the rounded values $\hat{a}_j$ satisfy $a_j\le\hat{a}_j\le a_j+\epsilon$. The first inequality makes every set that is feasible after rounding feasible for Eq.~\eqref{eq:closure}. The second makes every set that is feasible under the budget $\delta-(N+2)\epsilon$ feasible after rounding, since $|\mathcal{U}|\le N+2$. The program is exact on $N+2$ items and $\delta/\epsilon$ budget levels, which gives the cost bound and the time.\end{proof}

According to \textbf{Proposition}~\ref{prop:select}, $\epsilon=\delta/(10(N+2))$ loosens the budget by a tenth and keeps the time quadratic in $N$.

\section{Performance Evaluation}
\label{sec:exp}

\subsection{Experimental Setup}
\label{subsec:setup}

We simulate a federation of $N=50$ clients placed uniformly in a cell of radius 500 m around the edge server. The path loss is $128.1+37.6\log_{10}d$ dB with the distance $d$ in km~\cite{TR36814}, the fading is Rayleigh, and the bandwidth is 20 MHz. A link is down in a round when its rate falls below the broadcast rate, which is set so that about 20\% of the clients are offline on average. The teacher is Qwen2.5-7B-Instruct~\cite{Yang2024ARXIV}. Half of the clients hold same-family Qwen2.5-1.5B-Instruct students, and the other half hold cross-family Llama-3.2-1B-Instruct students. All models are fine-tuned with low-rank adapters~\cite{Hu2022ICLR}. The local data come from TOFU~\cite{Maini2024COLM}, whose 200 fictitious authors are split over the clients with four authors each, so no model has seen them in pre-training. Each round distills on 256 public prompts, and the fraction of prompts that touch $u$ is fixed before training. Training runs for $T=20$ rounds, and the request arrives at $t_u=10$. Client unlearning removes all authors of one client. Class unlearning removes a group of ten authors held by different clients, which plays the role of a class. Sample unlearning removes the samples of one client that carry a trigger phrase and a wrong answer. On the hardware testbed in Fig.~\ref{fig:testbed}, a Jetson AGX Thor serves as the edge server, and two Jetson AGX Orin boards with 64 GB of memory serve as a same-family and a cross-family client over a Wi-Fi LAN.

\begin{figure*}[!t]
\centering
\subfloat[Exposure and closure size.]{\includegraphics[width=0.24\textwidth]{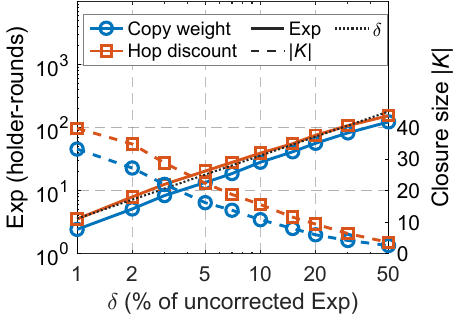}\label{fig:delta}}
\hfill
\subfloat[Resurrection and copy number.]{\includegraphics[width=0.24\textwidth]{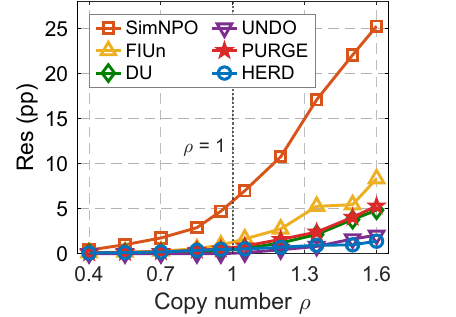}\label{fig:copynum}}
\hfill
\subfloat[Exposure under outages.]{\includegraphics[width=0.24\textwidth]{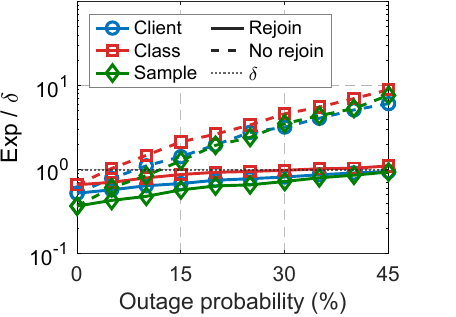}\label{fig:outage}}
\hfill
\subfloat[Edge estimates.]{\includegraphics[width=0.24\textwidth]{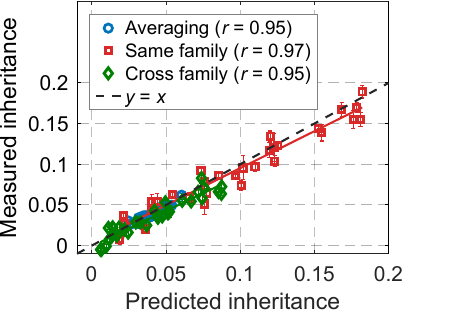}\label{fig:edge}}
\caption{Results of \method{}. (a) Cumulative exposure and closure size as $\delta$ varies, with copy weights and with hop-based weights. (b) Rise of the teacher's forget score after correction as $\rho(\bar{\mathbf{W}})$ varies. (c) Cumulative exposure with and without the rejoin rule as the outage probability grows. (d) Predicted versus measured inheritance of each averaging, same-family, and cross-family edge.}
\label{fig:results}
\end{figure*}

We measure unlearning with the forget ROUGE-L ($\mathrm{FRL}$), the ROUGE-L between the answers of a model on the forget set and the original answers, and with the TOFU model utility ($\mathrm{MU}$) on the retained authors. For sample unlearning, $\mathrm{FRL}$ is the share of triggered prompts that a model answers with the wrong answer. Both are reported for the teacher and averaged over the same-family students and over the cross-family students, and an unlearned holder should match a model retrained without $u$. For the teacher, we also report the TOFU forget quality ($\mathrm{FQ}$) and the balanced accuracy of a membership inference attack ($\mathrm{MIA}$) on the forget set, where 50\% is chance. $\mathrm{FQ}$ is the p-value of a Kolmogorov--Smirnov test between the truth ratios of the teacher and of the retrained teacher on the forget set. The resurrection ($\mathrm{Res}$) is the largest rise of the teacher's $\mathrm{FRL}$ within ten rounds after the correction. The exposure ($\mathrm{Exp}$) is the left side of $(\mathrm{C}_1)$, with each $r_j[t]$ read from the forget-set behavior of holder $j$. The cost is $\sum_{j\in\mathcal{K}}c_j$. On the testbed, we also measure the energy that the correction draws, read from the power monitors of the boards. We compare \method{} with five baselines. SimNPO~\cite{Fan2025NIPS} corrects only the teacher. FIUn~\cite{Liu2025TIFS} corrects every descendant of $u$ in an unweighted lineage graph. DU~\cite{Wu2026INFOCOM} sends the same holders a Newton-style correction, and UNDO~\cite{Lee2025NIPS} noises them and distills them again from the corrected teacher. PURGE~\cite{Quan2025NIPS} shards the data before training so that each shard of the teacher maps to one shard of the students.

\subsection{Experimental Results}
\label{subsec:results}

\begin{figure}[t]
\centering
\subfloat[Federation size.]{\includegraphics[width=0.485\linewidth]{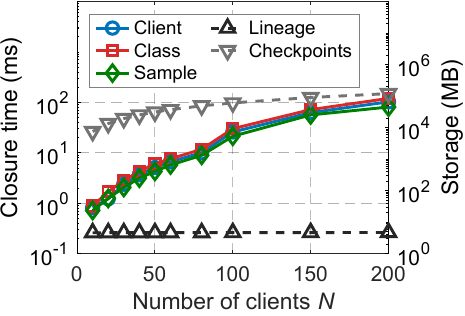}\label{fig:scale}}
\hfill
\subfloat[Request round.]{\includegraphics[width=0.485\linewidth]{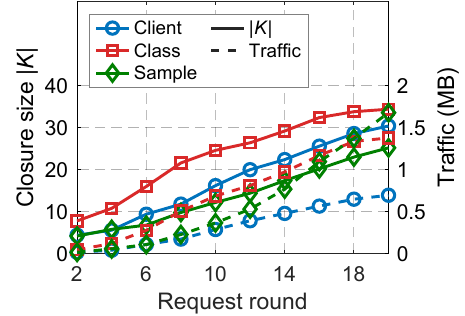}\label{fig:round}}
\caption{Scalability of \method{}. (a) Closure time and storage as $N$ grows. (b) Closure size and correction traffic as the request round $t_u$ grows.}
\label{fig:scalability}
\end{figure}

In Table~\ref{tab:main}, \method{} leaves 18\% to 52\% less exposure than UNDO, the strongest baseline, at 6\% to 12\% of its cost. In Fig.~\ref{fig:results}, the hop-based closure exceeds $\delta$ at seven of ten budgets. In Fig.~\ref{fig:scalability}, the closure takes at most 121~ms for $N=200$.

\section{Conclusion}
\label{sec:conclusion}

We studied federated unlearning under mutual distillation, where a removal executed on the teacher alone returns through the students. \method{} corrects the holders with the largest copy weights in a typed lineage graph and guards the rest.


\end{document}